\RequirePackage{fix-cm}
\documentclass[smallextended]{svjour3}       % onecolumn (second format)
\smartqed  % flush right qed marks, e.g. at end of proof
\usepackage{graphicx}
\usepackage{booktabs}
\usepackage{threeparttable}
\usepackage{graphicx}% Include figure files
\usepackage{dcolumn}% Align table columns on decimal point
\usepackage{bm,amsmath,amssymb}
\newtheorem{Theorem}{Theorem}
\newtheorem{Lemma}{Lemma}

\begin{document}

\title{Two New Families of Quantum Synchronizable Codes%\thanks{Grants or other notes
%about the article that should go on the front page should be
%placed here. General acknowledgments should be placed at the end of the article.}
}
%\subtitle{Do you have a subtitle?\\ If so, write it here}

%\titlerunning{Short form of title}        % if too long for running head

\author{Lan Luo         \and
        Zhi Ma  \and
        Dongdai Lin %etc.
}

%\authorrunning{Short form of author list} % if too long for running head

\institute{L. Lan \at
              Henan Key Laboratory of Network Cryptography Technology, Zhengzhou 450001, China\\
              State Key Laboratory of Information Security, Beijing 100093, China
              %Tel.: +123-45-678910\\
              %Fax: +123-45-678910\\
                        %  \\
%             \emph{Present address:} of F. Author  %  if needed
           \and
           Z. Ma \at
              Henan Key Laboratory of Network Cryptography Technology, Zhengzhou 450001, China\\
              \email{ma\_zhi@163.com}
           \and
           D. D. Lin\at
            State Key Laboratory of Information Security, Beijing 100093, China
}

\date{Received: date / Accepted: date}
% The correct dates will be entered by the editor

\maketitle

\begin{abstract}
In this paper, we present two new ways of quantum synchronization coding based on the $(\bm{u}+\bm{v}|\bm{u}-\bm{v})$ construction and the product construction respectively, and greatly enrich the varieties of available quantum synchronizable codes. The circumstances where the maximum synchronization error tolerance can be reached are explained for both constructions. Furthermore, repeated-root cyclic codes derived from the $(\bm{u}+\bm{v}|\bm{u}-\bm{v})$ construction are shown to be able to provide better Pauli error-correcting capability than BCH codes.

% \PACS{PACS code1 \and PACS code2 \and more}
% \subclass{MSC code1 \and MSC code2 \and more}
\end{abstract}

\section{\label{sec:level1}Introduction}

Block synchronization (or frame synchronization) is a critical problem in virtually any area in classical digital communications to ensure that the information transmitted can be correctly decoded by the receiver.  To achieve this goal, existing classical synchronization techniques commonly require that the information receiver or processing device constantly monitors the data to exactly identify the inserted boundary signals of an information block (see Refs.~\cite{Sklar2001Digital,Bregni2002Synchronization} for the basics of block synchronization techniques in classical digital communications).
Quantum block synchronization is also significant because the block structure is typically used in  quantum information coding~\cite{Nielsen2010Quantum,Lidar2013Quantum} as in classical domain and procedures for manipulating it demand precise alignment~\cite{Fujiwara2013High,Polyanskiy2013Asynchronous,Fujiwara2013Block}. However, since measurement of qubits usually destroys their contained quantum information, quantum analogues of above methods don't apply.

Aiming at this problem, Fujiwara~\cite{Fujiwara2013Block} proposed a solution---quantum synchronizable error-correcting codes, which allow us to eliminate the effects caused by block misalignment and Pauli errors. In his scheme, the construction of good quantum synchronizable codes demands a pair of nested dual-containing cyclic codes, both of which guarantee large minimum distances. Later, authors of Ref.~\cite{Fujiwara2013Algebraic} improved the original result by widening the range of tolerable magnitude of misalignment and presented several quantum synchronizable codes from classical BCH codes and punctured Reed-Muller (RM) codes. After that,  finite geometric codes~\cite{Fujiwara2014Quantum}, quadratic residue codes~\cite{Xie2014Quantum}, duadic codes~\cite{Guenda2015Algebraic} and repeated-root codes~\cite{xie2016Q,Lan2018Non} etc., were  shown to be applicable in synchronization coding.
However, apart from the case with repeated-root cyclic codes,  code parameters of other available quantum synchronizable codes are strongly limited~\cite{Lan2018Non}.
 Besides, the difficulty in computing the exact minimum distances of cyclic codes keeps us away from an accurate estimate on the error-correcting capability against Pauli errors of the obtained quantum codes.

In this work, we provide two new ways of constructing  quantum synchronizable codes. The first method exploits the well-known $(\bm{u}+\bm{v}|\bm{u}-\bm{v})$ construction
on cyclic codes and negacyclic codes  to generate new cyclic codes  with twice the lengths.
Two circumstances where the obtained quantum codes can achieve the maximum synchronization error tolerance are provided. In particular, repeated-root cyclic codes are shown
to be
able to provide better performance in correcting Pauli errors than non-primitive, narrow-sense BCH codes. The second method exploits the product construction to produce new cyclic codes from two cyclic codes with coprime lengths. With a broad range of  cyclic codes as ingredients, the varieties of quantum synchronizable codes are greatly extended  using cyclic product codes. Furthermore, the obtained codes can also reach the best attainable tolerance against misalignment under certain circumstances.

The rest of this paper is organized as follows: First we describe the general formalism of quantum synchronization coding in Section 2. Then we build quantum synchronizable codes based upon the $(\bm{u}+\bm{v}|\bm{u}-\bm{v})$ construction in Section 3.1. Two circumstances where the obtained quantum codes reach the maximum synchronization error tolerance are elaborated with repeated-root codes in Section 3.2. Afterwards, we discuss the minimum distances of above repeated-root  codes in Section 3.3. In Section 4, we produce quantum synchronizable codes from cyclic product codes. Finally, we give concluding remarks in Section 5.

\section{\label{sec:level1}Preliminaries}

In this section, we give a brief review of quantum synchronization coding. To start with, we describe some basic facts in classical and quantum coding theory. For further details, the readers can refer to Refs.~\cite{Nielsen2010Quantum,Huffman2010Fundamentals}.

Let $\mathbb{F}_{q}$ be a finite field where $q=p^m$ is a prime power. A classical $[n,k,d]$ linear code $C$ over $\mathbb{F}_{q}$ is a $k$-dimensional subspace of $\mathbb{F}_{q}^{n}$ such that $\text{min}\{\text{wt}(\bm{c}):\bm{c}\in C\backslash\{\bm{0}\}\}=d$, where  $\text{wt}(\bm{c})$ denotes the number of non-zero coordinates of a codeword $\bm{c}$.
$C$ can be determined as the null space of an $(n-k)\times k$ parity-check matrix $H$, i.e., $C=\{\bm{c}\in\mathbb{F}_{q}^{n}: H\bm{c}^{\mathsf{T}}=\bm{0}\}$. Accordingly, there exists a $k\times n$ generator matrix $G$ with its row space corresponding to $C$ such that $HG^{\mathsf{T}}=\bm{0}$. The dual code $C^{\perp}:=\{\bm{c}'\in\mathbb{F}_{q}^{n}:\bm{c}\bm{c}'^{\mathsf{T}}=0,\forall \bm{c}\in C\}$ is an $[n,n-k]$ code with a generator matrix $H$ and a parity-check matrix $G$. If $C\subset C^{\perp}$, we call $C$ a self-orthogonal code. Otherwise if $C^{\perp}\subset C$, $C$ is a dual-containing code.

A classical linear code $C$ is (nega)cyclic if it remains unchanged under a (nega)cyclic shifting of the coordinates, i.e., for a codeword $\bm{c}=(c_{0},c_{1},\dots,c_{n-1})\in C$, a cyclic shift $(c_{n-1},c_{0},\dots,c_{n-2})$ (a negacyclic shift $(-c_{n-1},c_{0},\dots,c_{n-2})$) is also a codeword of $C$. Especially, repeated-root (nega)cyclic codes~\cite{Dinh2008On,Chen2014Repeated,Dinh2013Structure,Chen2015Repeated,Ozadam2009The,Zeh2015Decoding} are those whose lengths are divisible by the characteristic $p$ of $\mathbb{F}_{q}$.
 Identify each codeword $\bm{c}$ as the coefficient vector of a polynomial $c(x)=\sum_{i=0}^{n-1}c_{i}x^{i}$. Then an $[n,k]$ (nega)cyclic code $C$ is equivalent with an ideal $\langle g(x)\rangle$ in the quotient ring $\frac{\mathbb{F}_{q}[x]}{\langle x^{n}-1\rangle}$ ($\frac{\mathbb{F}_{q}[x]}{\langle x^n+1\rangle}$). We call the monic polynomial $g(x)$ of degree $n-k$ as the generator polynomial of $C$.
If the value $\frac{q-1}{\text{gcd}(n,q-1)}$ is even,  there exists an isomorphism $\phi$ between the quotient rings $\frac{\mathbb{F}_{q}[x]}{\langle x^n-1\rangle}$ and $\frac{\mathbb{F}_{q}[x]}{\langle x^{n}+1\rangle}$ that maps $g(x)$ to $g(\lambda x)$~\cite{Chen2014Repeated}, where $\lambda^{n}=-1$.
Define the parity-check polynomial $h(x)$ of a (nega)cyclic code $C$ as $h(x)=\frac{x^n-1}{g(x)}$ ($h(x)=\frac{x^n+1}{g(x)}$).  Accordingly, the dual (nega)cyclic  code $C^{\perp}$ has a generator polynomial $g^{\perp}(x)=h^{*}(x)$, where $h^{*}(x)=h_{0}^{-1}x^{\text{deg}(h(x))}h(\frac{1}{x})$ represents the monic reciprocal polynomial of $h(x)$.

An $[[n,k,d]]$ quantum code $\mathcal{Q}$  is a $q^k$-dimensional subspace of a $q^n$-dimensional Hilbert space $\mathbb({C}^{q})^{\otimes n}$. Typically, $Q$ is designed to correct the errors caused by Pauli operators $X_{\bm{\alpha}}Z_{\bm{\beta}}$ of weight less than $\lfloor\frac{d-1}{2}\rfloor$, where $\bm{\alpha},\bm{\beta}\in\mathbb{F}_{q}^{n}$. An $(a_{l},a_{r})-[[n,k]]$ quantum synchronizable code is a quantum code that corrects not only Pauli errors, but also block misalignment to the left by $a_{l}$ qudits ($q$-ary quantum systems) and to the right by $a_{r}$ qudits for some non-negative integers $a_{l}$ and $a_{r}$. Denote the order of a polynomial $f(x)$ with $f(0)\not=0$ by $\text{ord}(f(x))$, i.e., $\text{ord}(f(x))=|\{x^a(\text{mod } f(x)):\ a\in \mathbb{N}\}|$. We give the quantum  synchronization coding framework as follows.

\begin{Theorem}~\cite{xie2016Q,Lan2018Non}
Let $C$ be a dual-containing $[n,k_{C},d_{C}]$ cyclic code and $D$ be an  $[n,k_{D},d_{D}]$ cyclic code such that $C\subset D$. Denote by $g_{C}(x)$ and $g_{D}(x)$ the generator polynomials of $C$ and $D$ respectively. Define the polynomial $f(x)$ of degree $k_{D}-k_{C}$ to be the quotient of $g_{C}(x)$ divided by $g_{D}(x)$. Then for any pair $a_{l},a_{r}$ of non-negative integers satisfying $a_{l}+a_{r}<\emph{ord}(f(x))$, there exists an  $(a_{l},a_{r})-[[n,2k_{C}-n]]$ quantum  synchronizable code that can correct up to $\lfloor\frac{d_{D}-1}{2}\rfloor$ bit errors and $\lfloor \frac{d_{C}-1}{2}\rfloor$ phase errors.
\label{thm1}
\end{Theorem}

 We can tell from Theorem~\ref{thm1} that a valid construction of good quantum synchronizable codes  relies on a pair of  dual-containing cyclic codes, one of which is contained in the other and both  guarantee large minimum distances.
Furthermore, the obtained synchronizable code  can correct synchronization errors (or misalignment) up to $\text{ord}(f(x))$ qubits. If $\text{ord}(f(x))=n$, then
the quantum synchronizable code achieves the maximum  synchronization error tolerance.

\section{\label{sec:level1}The $(\bm{u}+\bm{v}|\bm{u}-\bm{v})$ construction}
\subsection{\label{sec:level2}Synchronization coding}

In this section, we describe the quantum synchronization coding based upon the $(\bm{u}+\bm{v}|\bm{u}-\bm{v})$ construction.
  Compared with the well-known $(\bm{u}|\bm{u}+\bm{v})$ method---an iterative way to define RM codes,
the $(\bm{u}+\bm{v}|\bm{u}-\bm{v})$ technique
  has several advantages~\cite{Hughes2000Constacyclic,Ling2001On}. Apart from an estimate of minimum distances never worse than the other case, the $(\bm{u}+\bm{v}|\bm{u}-\bm{v})$ scheme enables us to obtain a $2n$-length cyclic code from an $n$-length cyclic code and an $n$-length negacyclic code.

To be specific, let $C_{1}$ and $C_{2}$ be $[n,k_{1},d_{1}]$ and $[n,k_{2},d_{2}]$ linear codes over $\mathbb{F}_{q}$ respectively, where $q=p^m$ is an odd prime power. (In this section, we leave the case with $p=2$ out of consideration.) Denote by $G_{1},G_{2}$ and $H_{1},H_{2}$ the generator matrices and parity-check matrices of $C_{1}$ and $C_{2}$, respectively.
The $(\bm{u}+\bm{v}|\bm{u}-\bm{v})$ construction
$C=C_{1}\curlyvee C_{2}=\{(\bm{u}+\bm{v}|\bm{u}-\bm{v}):\bm{u}\in C_{1},\bm{v}\in C_{2}\}$~\cite{Hughes2000Constacyclic,Macwilliams1977The} is a $[2n, k_{1}+k_{2}, \text{min}\{2d_{1},2d_{2},\text{max}\{d_{1},d_{2}\}\}]$ code with a generator matrix
\begin{equation}
G_{C}=\left(
\begin{array}{cc}
G_{1} & G_{1}\\
G_{2}  & -G_{2}
\end{array}
\right).
\end{equation}
%
%The dual code $C^{\perp}$ is also a $(\bm{u}+\bm{v}|\bm{u}-\bm{v})$ construction of $C_{1}^{\perp}$ and $C_{2}^{\perp}$, i.e., $C^{\perp}=(C_{1}\curlyvee C_{2})^{\perp}=C_{1}^{\perp}\curlyvee C_{2}^{\perp}$, with generator matrix
%\begin{equation}
%H=\left(
%\begin{array}{cc}
%H_{1} & H_{1}\\
%H_{2}  & -H_{2}
%\end{array}
%\right).
%\end{equation}
 Suppose that $C_{1}$ is cyclic with a generator polynomial $g_{1}(x)$ and $C_{2}$ is negacyclic with a generator polynomial $g_{2}(x)$, then $C$ is cyclic with a generator polynomial $g(x)=g_1(x)g_2(x)$~\cite{Hughes2000Constacyclic}.
 Clearly, $C$ is dual-containing if both $C_{1}$ and $C_{2}$ are dual-containing.  Applying these properties to Theorem~\ref{thm1}, we can build a family of quantum synchronizable codes from cyclic codes and negacyclic codes as follows.
\begin{Theorem}
Let $C_{i}$ be an $[n,k_{i},d_{i}]$ dual-containing code for $i\in\{1,2,3,4\}$. Suppose that $C_{1},C_{3}$ are cyclic with respective generator polynomial $g_{1}(x),g_{3}(x)$ and $C_{2},C_{4}$ are negacyclic with respective generator polynomial $g_{2}(x), g_{4}(x)$. If $C_{1}\subset C_{3}$ and $C_{2}\subset C_{4}$, define $f(x)=\frac{g_{1}(x)g_{2}(x)}{g_{3}(x)g_{4}(x)}$. Then for any pair of non-negative integers $a_{l},a_{r}$ such that $a_{l}+a_{r}<\emph{ord}(f(x))$, there exists an $(a_{l},a_{r})-[[2n,2(k_{1}+k_{2}-n)]]$ quantum synchronizable code that can correct up to $\lfloor \frac{\emph{min}\{2d_{3},2d_{4},\emph{max}\{d_{3},d_{4}\}\}-1}{2}\rfloor$ bit errors and $\lfloor \frac{\emph{min}\{2d_{1},2d_{2},\emph{max}\{d_{1},d_{2}\}\}-1}{2}\rfloor$ phase errors.
\label{thm2}
\end{Theorem}
\begin{proof}
It is clear that $C=C_{1}\curlyvee C_{2}$ is a $[2n,k_{1}+k_{2},\text{min}\{2d_{1},2d_{2},\text{max}\{d_{1},d_{2}\}\}]$ cyclic code with a generator polynomial $g_{C}(x)=g_{1}(x)g_{2}(x)$ and $D=C_{3}\curlyvee C_{4}$ is a $[2n,k_{3}+k_{4},\text{min}\{2d_{3},2d_{4},\text{max}\{d_{3},d_{4}\}\}]$ cyclic code with generator polynomial $g_{D}(x)=g_{3}(x)g_{4}(x)$. Furthermore, the condition $C\subset D$ holds because $C_{1}\subset C_{3}$ and $C_{2}\subset C_{4}$. By applying $C$ and $D$ to Theorem~\ref{thm1}, we can then obtain the required quantum synchronizable codes.
$\hfill\square$
\end{proof}

Different from Theorem~\ref{thm1}, Theorem~\ref{thm2} calls for
 two pairs of nested dual-containing classical codes in quantum synchronization coding, one of which are cyclic and the other are negacyclic.
 All of these codes need to guarantee large minimum distances, and  are desired to make as large $\text{ord}(f(x))$ as possible to offer better synchronization recovery capability. In particular, the maximum tolerable magnitude of misalignment   is $2n$. In that case, the quantum synchronizable codes from Theorem~\ref{thm2} can correct misalignment by up to $a_{l}$ qubits to the left and $a_{r}$ qubits to the right provided that $a_{l}+a_{r}<2n$.

 \subsection{Maximum synchronization error tolerance}

 Under two circumstances could the maximum synchronization error tolerance $2n$ be achieved, one of which is that   $\text{ord}(\frac{g_{1}(x)}{g_{3}(x)})=n$ and $\text{ord}(\frac{g_{2}(x)}{g_{4}(x)})=2$ where $\text{gcd}(n,2)=1$, and the other is that $\text{ord}(\frac{g_{2}(x)}{g_{4}(x)})=2n$ whatever the value of $\text{ord}(\frac{g_{1}(x)}{g_{3}(x)})$ is.

  \subsubsection{The first circumstance}
  The  condition  $\text{ord}(\frac{g_{1}(x)}{g_{3}(x)})=n$ in the first circumstance
 has been investigated on nearly all available quantum synchronizable codes, and is applicable to many cyclic codes, e.g., BCH codes~\cite{Fujiwara2013Algebraic}, punctured RM codes~\cite{Fujiwara2013Algebraic}, quadratic residue codes~\cite{Xie2014Quantum} and repeated-root cyclic codes~\cite{xie2016Q,Lan2018Non}. The other condition $\text{ord}(\frac{g_{2}(x)}{g_{4}(x)})=2$, however, has limited applications subject to the dual-containing constraint.
 One feasible solution is to use repeated-root codes of length $p^s$, where $s$ is a positive integer.

 To be concrete, let $C_{1},C_{3}$ be $p^s$-length dual-containing cyclic codes and $C_{2},C_{4}$ be $p^s$-length dual-containing negacyclic codes. Then $C_{1},C_{2},C_{3},C_{4}$ have  generator polynomials~\cite{Dinh2008On}
 \begin{equation}
 \begin{array}{l}
 g_{i}(x)=(x-1)^{p^s-k_{i}},\quad i=1,3,\\
 g_{j}(x)=(x+1)^{p^s-k_{j}},\quad j=2,4,\\
 \end{array}
 \end{equation}
 where $\frac{p^s}{2}\leq k_{i},k_{j}\leq p^s$. With the help of these codes, we can  build a family of   quantum synchronizable codes that possess the maximum synchronization error tolerance.
 \begin{Theorem}
 Let $C_{i}$ be a $[p^s,k_{i}]$ cyclic code and  $C_{j}$ be a $[p^s,k_{j}]$ negacyclic code, where $\frac{p^s}{2}\leq k_{i},k_{j}\leq p^s$ for $i\in\{1,3\}$ and $j\in\{2,4\}$. Suppose that $k_{3}-k_{1}>p^{s-1}$ and $k_{4}=k_{2}+1$, then for
 non-negative integers $a_{l}$ and $a_{r}$ such that $a_{l}+a_{r}<2p^s$, there exists an $(a_{l},a_{r})-[[2p^s,2(k_{1}+k_{3}-p^s)]]$ quantum synchronizable code.
  \end{Theorem}
  \begin{proof}
  The fact that $C_{1}\subset C_{3}$ and $C_{2}\subset C_{4}$ is evident since $k_{1}<k_{3}$ and $k_{2}<k_{4}$. Furthermore, the order of the polynomial
  $f(x)=\frac{g_{1}(x)g_{2}(x)}{g_{3}(x)g_{4}(x)}=(x+1)(x-1)^{k_{3}-k_{1}}$
  is $2p^s$. By applying these properties to Theorem~\ref{thm2}, we can naturally obtain  the quantum synchronizable codes of desired parameters.
  $\hfill\square$
  \end{proof}

  \subsubsection{The second circumstance}

 Assume that $\frac{q-1}{\text{gcd}(n,q-1)}$ is even, then there exists an isomorphism $\phi$ between the quotient rings $\frac{\mathbb{F}_{q}[x]}{\langle x^{n}-1\rangle}$ and $\frac{\mathbb{F}_{q}[x]}{\langle x^{n}+1\rangle}$ which maps $c(x)$ to $c(\lambda x)$, where $\lambda^{n}$=-1~\cite{Chen2014Repeated}. Furthermore, if the order of $c(x)$ is $n$, the order of $c(\lambda x)$ is $2n$. Therefore, the condition  $\text{ord}(\frac{g_{2}(x)}{g_{4}(x)})=2n$ on negacyclic codes $C_{2},C_{4}$
 can be achieved by finding two suitable cyclic codes.
 On that condition, most existing quantum synchronizable codes that provide the highest tolerance against synchronization errors can be generalized to quantum synchronizable codes of twice the lengths.
  As an example, we consider the use of $lp^s$-length repeated-root codes where $l$ is a prime distinct from $p$.

We first deal with the case $l\not=2$.
 Pick a primitive $l$-th root $\zeta$ of unity in the extension field $\mathbb{F}_{q^w}$ with $w=ord_{l}(q)$ indicating the order of $q$ in $\mathbb{Z}_{l}^{*}$.
 For  $0\leq t\leq e=\frac{l-1}{w}$, denote  by $M_{t}(x)$ the minimal polynomial of $\zeta^{t}$ over $\mathbb{F}_{q}$.
 The following lemma describes the structures of $lp^s$-length cyclic codes and negacyclic codes explicitly.

  %\begin{Lemma}~\cite{Chen2014Repeated}
%  Suppose that
%  $C=\langle g(x)\rangle\subset\frac{\mathbb{F}_{q}[x]}{\langle x^{lp^s}-1\rangle} $.
%  \begin{itemize}
%  \item[(1).] If $\emph{gcd}(l,q-1)=1$, then
%  $g(x)= \prod_{t=0}^{e}(M_{t}(x))^{p^s-a_{t}}$
%   with $0\leq a_{t}\leq p^s$.
%      Especially when $w=\emph{ord}_{l}(q)$ is odd, $g(x)$ can be represented as
%      $g(x)=(x-1)^{p^s-a_{0}}\prod_{t=1}^{\frac{e}{2}}(M_{t}(x))^{p^s-a_{t}}(M_{-t}(x))^{p^s-a_{-t}}$.
%         \item[(2).] If $\emph{gcd}(l,q-1)=l$, there exists a primitive $l$-th root $\zeta$ of unity in $\mathbb{F}_{q}$. Hence we can get
%       $g(x)=(x-1)^{p^s-a_{0}}\prod_{t=1}^{\frac{l-1}{2}}(x-\zeta^{t})^{p^s-a_{t}}(x-\zeta^{-t})^{p^s-a_{-t}}$
%       with $0\leq a_{t},a_{-t}\leq p^s$ for $0\leq t\leq \frac{l-1}{2}$.
%  \end{itemize}
%  \end{Lemma}

  \begin{Lemma}~\cite{Chen2014Repeated}
 Let $C_{1}$ be an $lp^s$-length cyclic code with a generator polynomial $g_{1}(x)$ and let $C_{2}$ be an $lp^s$-length negacyclic code with a generator polynomial $g_{2}(x)$.
  \begin{itemize}
  \item[(I).]
  If $\emph{gcd}(l,q-1)=1$, then
  \begin{equation}
  \begin{array}{l}
  g_{1}(x)= \prod_{t=0}^{e}(M_{t}(x))^{p^s-a_{1,t}},\\
  g_{2}(x)=\prod_{t=0}^{e}(\hat{M}_{t}(-x))^{p^{s}-a_{2,t}},
  \end{array}
  \end{equation}
  where $0\leq a_{1,t},a_{2,t}\leq p^s$ for all $t$.
  The notation $\hat{M}_{t}(x)$ denotes the monic polynomial of $M_{t}(x)$ dividing its leading coefficient. In particular when $w=ord_{l}(q)$ is odd, the generator polynomials of $C_{1}$ and $C_{2}$ are given by   \begin{equation}
  \begin{array}{l}
  g_{1}(x)=(x-1)^{p^s-a_{1,0}}\prod_{t=1}^{\frac{e}{2}}(M_{t}(x))^{p^s-a_{1,t}}(M_{-t}(x))^{p^s-a_{1,-t}},\\
  g_{2}(x)=(x+1)^{p^s-a_{2,0}}\prod_{t=1}^{\frac{e}{2}}(\hat{M}_{t}(-x))^{p^{s}-a_{2,t}}(\hat{M}_{-t}(-x))^{p^s-a_{2,-t}},
\end{array}
  \end{equation}
  where  $0\leq a_{1,t},a_{2,t},a_{1,-t},a_{2,-t}\leq p^s$ for $0\leq t\leq \frac{e}{2}$.

      Correspondingly, if $w$ is even, the dual codes $C_{1}^{\perp}$ and $C_{2}^{\perp}$ have generator polynomials
  \begin{equation}
  \begin{array}{l}
  g_{1}^{\perp}(x)=\prod_{t=0}^{e}(M_{t}(x))^{a_{1,t}},\\
  g_{2}^{\perp}(x)=\prod_{t=0}^{e}(\hat{M}_{t}(-x))^{a_{2,t}},
  \end{array}
  \end{equation}
 respectively. Otherwise if $w$ is odd, the dual codes have respective generator polynomial
 \begin{equation}
 \begin{array}{l}
 g_{1}^{\perp}(x)=(x-1)^{a_{1,0}}\prod_{t=1}^{\frac{e}{2}}(M_{t}(x))^{a_{1,-t}}(M_{-t}(x))^{a_{1,t}},\\
 g_{2}^{\perp}(x)=(x+1)^{a_{2,0}}\prod_{t=1}^{\frac{e}{2}}(\hat{M}_{t}(-x))^{a_{2,-t}}(\hat{M}_{-t}(-x))^{a_{2,t}}.
 \end{array}
 \end{equation}

 \item[(II).] If $\emph{gcd}(l,q-1)=l$, then we have
        \begin{equation}
        \begin{array}{l}
        g_{1}(x)=(x-1)^{p^s-a_{1,0}}\prod_{t=1}^{\frac{l-1}{2}}(x-\zeta^{t})^{p^s-a_{1,t}}(x-\zeta^{-t})^{p^s-a_{1,-t}},\\
        g_{2}(x)=(x+1)^{p^s-a_{2,0}}\prod_{t=1}^{\frac{l-1}{2}}(x+\zeta^{t})^{p^s-a_{2,t}}(x+\zeta^{-t})^{p^s-a_{2,-t}},
        \end{array}
        \end{equation}
       where $0\leq a_{1,t},a_{2,t},a_{1,-t},a_{2,-t}\leq p^s$ for $0\leq t\leq \frac{l-1}{2}$. The dual codes $C_{1}^{\perp}$ and $C_{2}^{\perp}$ have generator polynomials
       \begin{equation}
       \begin{array}{l}
       g_{1}^{\perp}(x)=(x-1)^{a_{1,0}}\prod_{t=1}^{\frac{l-1}{2}}(x-\zeta^{t})^{a_{1,-t}}(x-\zeta^{-t})^{a_{1,t}},\\
       g_{2}^{\perp}(x)=(x+1)^{a_{2,0}}\prod_{t=1}^{\frac{l-1}{2}}(x+\zeta^{t})^{a_{2,-t}}(x+\zeta^{-t})^{a_{2,t}},
       \end{array}
       \end{equation}
       respectively.
 \end{itemize}
\end{Lemma}

By applying above codes to Theorem~\ref{thm2}, we can build
 quantum synchronizable codes of length $2lp^s$ as follows.
  \begin{Theorem}
 Let $l$ be an odd prime such that $\emph{gcd}(l,q-1)=1$. Suppose that  $C_{1}, C_{3}$ are dual-containing cyclic codes of length $lp^s$  and $C_{2}, C_{4}$ are dual-containing negacyclic codes of length $lp^s$.
   \begin{itemize}
  \item[(I).] If $w$ is even, then
   $C_{i},C_{j}$ have generator polynomials
  \begin{equation}
  \begin{array}{l}
  g_{i}(x)=\prod_{t=0}^{e}(M_{t}(x))^{p^s-a_{i,t}},\quad i\in\{1,3\},\\
  g_{j}(x)=\prod_{t=0}^{e}(\hat{M}_{t}(-x))^{p^s-a_{j,t}},\quad j\in\{2,4\},
  \end{array}
  \end{equation}
  respectively, where $\frac{p^s}{2}\leq a_{i,t},a_{j,t}\leq p^s$  for $0\leq t\leq e$.
  Assume that $a_{1,t}\leq a_{3,t}$ and $a_{2,t}\leq a_{4,t}$ for all $t$. If there exists an integer $r$ in the range $0\leq r\leq e$ such that $\emph{gcd}(r,l)=1$ and $a_{4,r}-a_{2,r}>p^{s-1}$, then we can  construct an $(a_{l},a_{r})-[[2lp^s,k]]$
      quantum synchronizable code where
      \begin{equation}
      k=2\left(\sum_{t=1}^{e}(a_{1,t}+a_{2,t})w+(a_{1,0}+a_{2,0})-lp^s\right).
      \end{equation}
  \item[(II).]
     If $w$ is odd, then
      $C_{i}$ and $C_{j}$ have generator polynomials
      \begin{equation}
      \begin{array}{l}
      g_{i}(x)=(x-1)^{p^s-a_{i,0}}\prod_{t=1}^{\frac{e}{2}}(M_{t}(x))^{p^s-a_{i,t}}(M_{-t}(x))^{p^{s}-a_{i,-t}},\quad i\in\{1,3\},\\
      g_{j}(x)=(x+1)^{p^s-a_{j,0}}\prod_{t=1}^{\frac{e}{2}}(\hat{M}_{t}(-x))^{p^s-a_{j,t}}(\hat{M}_{-t}(-x))^{p^s-a_{j,-t}},\quad j\in\{2,4\},
      \end{array}
      \end{equation}
      respectively, where $\frac{p^s}{2}\leq a_{i,0},a_{j,0}\leq p^s$ and $p^s\leq a_{i,t}+a_{i,-t},a_{j,t}+a_{j,-t}\leq 2p^s$ for $1\leq t\leq \frac{e}{2}$. Assume that
      \begin{equation}
      \begin{array}{ll}
         a_{1,0}\leq a_{3,0},& a_{2,0}\leq a_{4,0},\\
       a_{1,t}\leq a_{3,t},& a_{2,t}\leq a_{4,t},\\
       a_{1,-t}\leq a_{3,-t},& a_{2,-t}\leq a_{4,-t},
         \end{array}
         \end{equation}
      for $1\leq t\leq \frac{e}{2}$.
      If there exists an integer $r$ in the range $-\frac{e}{2} \leq r\leq \frac{e}{2}$ such that $\emph{gcd}(r,l)=1$ and $a_{4,r}-a_{2,r}>p^{s-1}$, then for any non-negative integers $a_{l},a_{r}$ satisfying $a_{l}+a_{r}<2lp^s$, we can obtain an $(a_{l},a_{r})-[[2lp^s,k]]$ quantum synchronizable code where
      \begin{equation}
      k=2\left(\sum\limits_{1\leq t\leq \frac{e}{2}}(a_{1,t}+a_{1,-t}+a_{2,t}+a_{2,-t})w+(a_{1,0}+a_{2,0})-lp^s\right).
      \end{equation}
  \end{itemize}
  \label{thm4}
  \end{Theorem}
  \begin{proof}
 Given an even $w$, the dual-containing properties of $C_{i}$ and $C_{j}$, for $i\in\{1,3\}$ and $j\in\{2,4\}$,  are guaranteed when the parameters $a_{i,t},a_{j,t}$ are in the range $\frac{p^s}{2}\leq a_{i,t},a_{j,t}\leq p^s$ for all $t$. Furthermore, due to the assumption that $a_{1,t}\leq a_{3,t}$ and $a_{2,t}\leq a_{4,t}$ for $0\leq t\leq e$, we have $C_{1}\subset C_{3}$ and $C_{2}\subset C_{4}$. In that case, the polynomial $f(x)$ in Theorem~\ref{thm2} is
 \begin{equation}
 \begin{array}{ll}
 f(x)&=\frac{g_{1}(x)g_{2}(x)}{g_{3}(x)g_{4}(x)}=\frac{\prod_{t=0}^{e}(M_{t}(x))^{p^s-a_{1,t}}(\hat{M}_{t}(-x))^{p^s-a_{2,t}}}{\prod_{t=0}^{e}(M_{t}(x))^{p^s-a_{3,t}}(\hat{M}_{t}(-x))^{p^s-a_{4,t}}}\\
 & =\prod_{t=0}^{e}(M_{t}(x))^{a_{3,t}-a_{1,t}}(\hat{M}_{t}(-x))^{a_{4,t}-a_{2,t}}.
 \end{array}
 \end{equation}
 Pick an integer $r$ that is relatively prime to $l$ such that $a_{4,r}-a_{2,r}>p^{s-1}$, then the order of $f(x)$ has a factor $p^s\cdot\text{ord}(\hat{M}_{r}(-x))$. Note that $\text{ord}(\hat{M}_{r}(x))=\frac{l}{\text{gcd}(r,l)}=l$. Hence we have $\text{ord}(\hat{M}_{r}(-x))=2l$, indicating that $\text{ord}(f(x))\geq 2lp^s$. Due to the fact that $\text{ord}(f(x))\leq 2lp^s$, we can finally obtain that $\text{ord}(f(x))=2lp^s$. Moreover, $C_{i}$ has dimension
 \begin{equation}
 k_{i}=lp^s-\left((p^s-a_{i,0})+\sum_{t=1}^{e}(p^s-a_{i,t})w\right)=a_{i,0}+\sum_{t=1}^{e}a_{i,t}\cdot w.
 \end{equation}
And analogously, $C_{j}$ has dimension $k_{j}=a_{j,0}+\sum_{t=1}^{e}a_{j,t}\cdot w$. Therefore, the quantum synchronizable code built on them has the desired parameters.
For an odd $w$, the statements in (II) can be verified using similar arguments.
  $\hfill\square$
  \end{proof}

 \begin{Theorem}
Let $l$ be an odd prime such that $\text{gcd}(l,q-1)=l$. Suppose that $C_{1},C_{3}$ are dual-containing cyclic codes of length $lp^s$ and $C_{2},C_{4}$ are dual-containing negacyclic codes of length $lp^s$. The generator polynomials of $C_{i},C_{j}$ for $i\in\{1,3\}$ and $j\in\{2,4\}$ are
\begin{equation}
\begin{array}{l}
g_{i}(x)=(x-1)^{p^s-a_{i,0}}\prod_{t=1}^{\frac{l-1}{2}}(x-\zeta^{t})^{p^{s}-a_{i,t}}(x-\zeta^{-t})^{p^s-a_{i,-t}},\\
g_{j}(x)=(x+1)^{p^s-a_{j,0}}\prod_{t=1}^{\frac{l-1}{2}}(x+\zeta^{t})^{p^s-a_{j,t}}(x+\zeta^{-t})^{p^s-a_{j,-t}},
\end{array}
\end{equation}
respectively, where $\frac{p^s}{2}\leq a_{i,0},a_{j,0}\leq p^s$ and $p^s\leq a_{i,t}+a_{i,-t},a_{j,t}+a_{j,-t}\leq 2p^s$ for $1\leq t\leq \frac{l-1}{2}$. Assume that
\begin{equation}
\begin{array}{ll}
a_{1,0}\leq a_{3,0},& a_{2,0}\leq a_{4,0},\\
a_{1,t}\leq a_{3,t},& a_{2,t}\leq a_{4,t},\\
a_{1,-t}\leq a_{3,-t},& a_{2,-t}\leq a_{4,-t},
\end{array}
\end{equation}
for $1\leq t\leq \frac{l-1}{2}$. If we can pick an integer $r$ with $-\frac{l-1}{2}\leq r\leq \frac{l-1}{2}$ such that $\text{gcd}(r,l)=1$ and $a_{4,r}-a_{2,r}>p^{s-1}$, then given a pair of non-negative integers $a_{l},a_{r}$ satisfying $a_{l}+a_{r}<2lp^s$, there exists
a quantum synchronizable code of length $2lp^s$ and dimension
\begin{equation}
k=2\left((a_{1,0}+a_{2,0})+\sum_{t=1}^{\frac{l-1}{2}}(a_{1,t}+a_{1,-t}+a_{2,t}+a_{2,-t})-lp^s\right).
\end{equation}
\label{thm5}
\end{Theorem}
\begin{proof}
Following a  similar proof to that of Theorem~\ref{thm4}, we can obtain the desired results.
$\hfill\square$
\end{proof}

In the case of $l=2$, quantum synchronizable codes that possess the maximum synchronization error tolerance can also be constructed from $lp^s$-length cyclic codes and negacyclic codes. The following lemma describes the structures of $2p^s$-length repeated-root codes clearly.
\begin{Lemma}~\cite{Chen2014Repeated}
Let $C_{1}$ be a cyclic code of length $2p^s$. Then $C_{1}$ and its dual code $C_{1}^{\perp}$ have  respective generator polynomial
\begin{equation}
\begin{array}{l}
g_{1}(x)=(x-1)^{p^s-a_{1,0}}(x+1)^{p^s-a_{1,1}},\\
g_{1}^{\perp}(x)=(x-1)^{a_{1,0}}(x+1)^{a_{1,1}},
\end{array}
\end{equation}
where $0\leq a_{1,0},a_{1,1}\leq p^s$.

Let $C_{2}$ be a negacyclic code of length $2p^s$.
If $q\equiv 1(\emph{mod } 4)$, there exists an element $\eta\in\mathbb{F}_{q}^{*}$ such that $\eta^{2}=-1$. In that case,
$C_{2}$ and $C_{2}^{\perp}$ have generator polynomials
\begin{equation}
\begin{array}{l}
g_{2}(x)=(x-\eta)^{p^s-a_{2,0}}(x+\eta)^{p^s-a_{2,1}},\\
g_{2}^{\perp}(x)=(x-\eta)^{a_{2,1}}(x+\eta)^{a_{2,0}},
\end{array}
\end{equation}
respectively, where $0\leq a_{2,0},a_{2,1}\leq p^s$. Otherwise if $q\equiv 3(\emph{mod } 4)$, their generator polynomials are given by
\begin{equation}
\begin{array}{l}
g_{2}(x)=(x^2+1)^{p^s-a_{2}},\\
g_{2}^{\perp}(x)=(x^2+1)^{a_{2}},
\end{array}
\end{equation}
where $0\leq a_{2}\leq p^s$.
\label{lem3}
\end{Lemma}

Taking similar arguments to the proof of Theorem~\ref{thm4}, we can obtain the following results.

\begin{Theorem}
Suppose that $C_{1},C_{3}$ are dual-containing cyclic codes of length $2p^s$ and $C_{2},C_{4}$ are dual-containing negacyclic codes of length $2p^s$.
\begin{itemize}
\item[(I).] If $q\equiv 1(\emph{mod }4)$, the generator polynomials of $C_{1},C_{2},C_{3},C_{4}$ are of the forms
\begin{equation}
\begin{array}{ll}
g_{i}(x)=(x-1)^{p^s-a_{i,0}}(x+1)^{p^s-a_{i,1}}, & i\in\{1,3\},\\
g_{j}(x)=(x-\eta)^{p^s-a_{j,0}}(x+\eta)^{p^s-a_{j,1}}, & j\in\{2,4\},
\end{array}
\end{equation}
where $\frac{p^s}{2}\leq a_{i,0},a_{i,1}\leq p^s$ and $p^s\leq a_{j,0}+a_{j,1}\leq 2p^s$. The notation $\eta$ denotes the element in $\mathbb{F}_{q}^{*}$ such that $\eta^{2}=-1$.
Assume that
\begin{equation}
\begin{array}{ll}
a_{1,0}\leq a_{3,0}, & a_{1,1}\leq a_{3,1},\\
a_{2,0}\leq a_{4,0}, & a_{2,1}\leq a_{4,1}.
\end{array}
\end{equation}
If either $a_{4,0}-a_{2,0}>p^{s-1}$ or $a_{4,1}-a_{2,1}>p^{s-1}$ holds, then for any non-negative pair $a_{l},a_{r}$ satisfying $a_{l}+a_{r}<4p^s$, we can build an $(a_{l},a_{r})-[[4p^s,k]]$ quantum synchronizable code, where
\begin{equation}
k=2\left(a_{1,0}+a_{1,1}+a_{2,0}+a_{2,1}-2p^s\right).
\end{equation}

\item[(II).] If $q\equiv 3(\emph{mod } 4)$, the generator polynomials of $C_{1},C_{2},C_{3},C_{4}$ are given by
    \begin{equation}
    \begin{array}{ll}
    g_{i}(x)=(x-1)^{p^s-a_{i,0}}(x+1)^{p^s-a_{i,1}},& i\in\{1,3\},\\
    g_{j}(x)=(x^2+1)^{p^s-a_{j}},& j\in\{2,4\},
    \end{array}
    \end{equation}
    where $\frac{p^s}{2}\leq a_{i,0},a_{i,2},a_{j}\leq p^s$. Assume that $a_{1,0}\leq a_{3,0}$, $a_{1,1}\leq a_{3,1}$ and $a_{2}\leq a_{4}$. If $a_{4}-a_{2}>p^{s-1}$, then for any non-negative integers $a_{l},a_{r}$ such that $a_{l}+a_{r}<4p^s$, there exists an $(a_{l},a_{r})-[[4p^s,k]]$ quantum synchronizable code, where
    \begin{equation}
    k=2\left(a_{1,0}+a_{1,1}+2a_{2}-2p^s\right).
    \end{equation}
\end{itemize}
\label{thm6}
\end{Theorem}

We can tell from Theorems~\ref{thm4},~\ref{thm5} and~\ref{thm6} that  quantum synchronizable codes of length $2lp^s$ that reach the maximum synchronization error tolerance can be derived from cyclic codes and negacyclic codes of length $lp^s$. This can be seen as a generalization of  results in Ref.~\cite{Lan2018Non}, where $lp^s$-length cyclic codes are exploited in the construction of
$lp^s$-length quantum synchronizable codes that tolerate misalignment by $lp^s$ qubits. Similar generalizations can be applied to other existing quantum synchronizable codes of length $n$ due to the isomorphism $\phi$ between $\frac{\mathbb{F}_{q}[x]}{\langle x^n-1\rangle}$ and $\frac{\mathbb{F}_{q}[x]}{\langle x^n+1\rangle}$~\cite{Chen2014Repeated} that maps $c(x)$ to $c(\lambda x)$ where $\lambda^n=-1$, when $\frac{q-1}{\text{gcd}(n,q-1)}$ is even.

\subsection{\label{sec:level2}The minimum distances}

From Theorem~\ref{thm2} we can see that, quantum synchronizable codes derived from cyclic codes $C_{1}\curlyvee C_{2}$ and $C_{3}\curlyvee C_{4}$ have minimum distances no worse or up to twice larger than those from the component cyclic codes $C_{1}$ and $C_{3}$. In other words, quantum synchronizable codes based on the $(\bm{u}+\bm{v}|\bm{u}-\bm{v})$ scheme can provide good performance in correcting Pauli errors. Take the codes from Theorem~\ref{thm4} (I) for an example.

Suppose that $l$ is an odd prime such that $\text{gcd}(l,q-1)$. If $w$ is even, then an $lp^s$-length cyclic code $C_{i}$ and an $lp^s$-length negacyclic code $C_{j}$, for $i\in\{1,3\}$ and $j\in\{2,4\}$, have respective generator polynomial
\begin{equation}
\begin{array}{l}
g_{i}(x)=\prod_{t=0}^{e}(M_{t}(x))^{p^s-a_{i,t}},\\
g_{j}(x)=\prod_{t=0}^{e}(\hat{M}_{t}(-x))^{p^s-a_{j,t}},
\end{array}
\end{equation}
where $0\leq a_{i,t},a_{j,t}\leq p^s$ for all $t$. The minimum distance of $C_{i}$ has been thoroughly investigated in Ref.~\cite{Lan2018Non}. Due to the isomorphism $\phi$ between $\frac{\mathbb{F}_{q}[x]}{\langle x^{lp^s}-1\rangle}$ and $\frac{\mathbb{F}_{q}[x]}{\langle x^{lp^s}+1\rangle}$ that maps $c(x)$ to $c(-x)$, the minimum distance of the negacyclic code $C_{j}$ can be computed using the same strategies.

Define a set of $l$-length negacyclic codes $\{\overline{C}_{j,v}:0\leq v\leq p^s-1\}$ with respective generator polynomial $\overline{g}_{j,v}(x)=\prod_{t=0}^{e}(\hat{M}_{t}(-x))^{f_{v,a_{j,t}}}$,
where $f_{v,a_{j,t}}=\left\{\begin{array}{ll}1,& p^s-a_{j,t}>v,\\ 0, & \text{otherwise.}\end{array}\right.$ For $0\leq v\leq p^s-1$, denote by
$P_{v}$  the Hamming weight of the polynomial $(x-1)^{v}$~\cite{Castagnoli1991On} and define the set
\begin{equation}
V=\left\{\sum_{\mu=1}^{u-1}(p-1)p^{s-\mu}+\tau p^{s-u}:1\leq u\leq s, 1
\leq \tau\leq p-1\right\}\cup\{0\}.
\end{equation}
The minimum distance of $C_{j}$ is demonstrated in Table~\ref{tab1}.

\begin{table*}[tp]
  \scriptsize
  \caption{ The minimum distance  of an $lp^s$-length negacyclic code $C_{j}=\langle \prod_{t=0}^{e}(\hat{M}_{t}(-x))^{p^s-a_{j,t}}\rangle$ with $0\leq a_{j,t}\leq p^s$ for $0\leq t\leq e$ and $j\in\{2,4\}$.}
  \setlength{\belowcaptionskip}{6pt}
  \renewcommand\arraystretch{1.3}
  \begin{threeparttable}
  \begin{tabular}{p{0.5cm}<{\centering}|p{3cm}<{\centering}|p{3cm}<{\centering}|p{3cm}<{\centering}}
  \toprule
  \textbf{Case} & $\bm{b_{\rm{min}}}^{\dagger}$ & $\bm{b_{\rm{max}}}$ & \textbf{minimum distance}\\
  \midrule
  1  &  $(0, p^{s-1}]$ & $(0, p^{s-1}]$ & 2\\
  \hline
  2  &  $(0, p^{s-1}]$ & $(\beta^{\star} p^{s-1},(\beta+1)p^{s-1}]$  & $\text{min}\{2d'^{\ddagger},\beta+2\}$\\
  \hline
  3  &  $(0, p^{s-1}]$ & $(p^s-p^{s-\mu}+(\tau-1)p^{s-\mu-1}, p^s-p^{s-\mu}+\tau p^{s-\mu-1}]$ & $\text{min}\{2d',(\tau+1)p^{\mu}\}$\\
  \hline
     4&  $(\beta_{0}p^{s-1},(\beta_{0}+1)p^{s-1}]$ & $(\beta_{1}p^{s-1},(\beta_{1}+1)p^{s-1}]$ & $\text{min}\{(\beta_{0}+2)d',\beta_{1}+2\}$\\
  \hline
  5  &  $(\beta p^{s-1},(\beta+1)p^{s-1}]$ & $(p^s-p^{s-\mu}+(\tau-1)p^{s-\mu-1},p^s-p^{s-\mu}+\tau p^{s-\mu-1}]$ & $\text{min}\{(\beta+2)d',(\tau+1)p^{\mu}\}$\\
  \hline
  6  &  $(p^s-p^{s-\mu_{0}}+(\tau_{0}-1)p^{s-\mu_{0}-1},p^s-p^{s-u_{0}}+\tau_{0}p^{s-\mu_{0}-1}]$ & $(p^{s}-p^{s-\mu_{1}}+(\tau_{1}-1)p^{s-\mu_{1}-1},p^{s}-p^{s-\mu_{1}}+\tau_{1}p^{s-\mu_{1}-1}]$ & $\text{min}\{(\tau_{0}+1)p^{\mu_{0}}d',(\tau_{1}+1)p^{\mu_{1}}\}$\\
  \hline
  7 & $(0,p^{s-1}]$ & $p^s$ & $2d'$\\
  \hline
  8 & $(\beta p^{s-1},(\beta+1)p^{s-1}]$ & $p^s$ & $(\beta+2)d'$\\
  \hline
  9 & $(p^s-p^{s-\mu}+(\tau-1)p^{s-\mu-1},p^s-p^{s-\mu}+\tau p^{s-\mu-1}]$ & $p^s$ & $(\tau+1)p^{\mu}d'$\\
  \bottomrule
  \end{tabular}
  \begin{tablenotes}
  \item[$\dagger$] The notations $b_{\text{min}}$ and $b_{\text{max}}$ denote the minimum and maximum elements in the set $\{p^s-a_{j,t}:0\leq t\leq e\}$, respectively.
       \item[$\star$] The parameters are in the ranges
      $1\leq\beta,\beta_{0},\beta_{1}\leq p-2$, $1\leq \mu,\mu_{0},\mu_{1}\leq s-1$ and $1\leq \tau,\tau_{0},\tau_{1}\leq p-1$.
    \item[$\ddagger$] $d'$ denotes the minimum distance of $\overline{C}_{j,v'}$ where $P_{v'}=\text{min}\{P_{v}:b_{\text{min}}\leq v< b_{\text{max}},v\in V\}$.

  \end{tablenotes}
  \end{threeparttable}
  \label{tab1}
\end{table*}

Furthermore, set $l=3$ and $q\equiv 2(\text{mod } 3)$. Thus a $3p^s$-length negacyclic code
$C_{j}$ has a generator polynomial
\begin{equation}
%g_{1}(x)=(x-1)^{p^s-a_{1,0}}(x^2+x+1)^{p^s-a_{1,1}},\\
g_{j}(x)=(x+1)^{p^s-a_{j,0}}(x^2-x+1)^{p^s-a_{j,1}},
\end{equation}
where $0\leq a_{j,0},a_{j,1}\leq p^s$.
The $l$-length negacyclic code $\overline{C}_{j,v}$ has minimum distance
\begin{equation}
d(\overline{C}_{j,v})=\left\{
\begin{array}{ll}
1, &\text{if } p^s-a_{j,0}\leq v, p^s-a_{j,1}\leq v,\\
2, &\text{if } p^s-a_{j,0}>v, p^s-a_{j,1}\leq v,\\
3, &\text{if } p^s-a_{j,0}\leq v, p^s-a_{j,1}>v,\\
\infty, &\text{if } p^s-a_{j,0}>v, p^s-a_{j,1}>v.
\end{array}
\right.
\end{equation}
Hence if we assume that $p^s-a_{j,0}< p^s-a_{j,1}$, the minimum distance of $\overline{C}_{j,v'}$ is 3,  where $P_{v'}=\text{min}\{p^s-a_{j,0}\leq v<p^s-a_{j,1} : v\in V\}$. On that condition, Table~\ref{tab2} lists  sample parameters for $C_{j}$.

\begin{table}[tp]
  \scriptsize
  \centering
  \caption{Sample parameters for an $[n_{j},k_{j},d_{j}]$ negacyclic code $C_{j}=\langle (x+1)^{p^s-a_{j,0}}(x^2-x+1)^{p^s-a_{j,1}}\rangle$ with $a_{j,1}<a_{j,0}$ for $j\in\{2,4\}$.}
  \setlength{\belowcaptionskip}{6pt}
  \renewcommand\arraystretch{1.3}
  \begin{threeparttable}
  \begin{tabular}{p{0.9cm}<{\centering}p{0.9cm}<{\centering}p{0.9cm}<{\centering}p{0.9cm}<{\centering}p{0.9cm}<{\centering}p{0.9cm}<{\centering}p{0.9cm}<{\centering}}
  \toprule
  $\bm{p}$ & $\bm{s}$ & $n_{j}$ & $\bm{a_{j,0}}$ & $\bm{a_{j,1}}$ &  $k_{j}$ & $\bm{d_{j}}$\\
  \midrule
  5 & 2 & 75 & 19 & 18 & 55 & 3\\
  5 & 2 & 75 & 19 & 4 &  27 & 9 \\
  5 & 2 & 75 & 14 & 3 &  20 & 12 \\
  5 & 2 & 75 & 4 & 2 & 8 & 20 \\
  5 & 3 & 375 & 14 & 4 & 22& 50 \\
  5 & 3 & 375 & 9 & 3 & 15 & 75 \\
  5 & 4 & 1875 & 19 & 4 & 27 & 225 \\
  5 & 4 & 1875 & 9 & 3 & 15 & 375 \\
  11& 2 & 363 & 109 & 65 & 239 & 7 \\
  11& 2 & 363 & 21 & 10 & 41 & 33 \\
  11 & 2 & 363 & 10 & 6 & 22 & 66 \\
  11 & 3 & 3993 & 65 & 10 & 85 & 231\\
  11 & 3 & 3993 & 32 & 9 & 50 & 330\\
  23 & 2 & 1587 & 459 & 275 & 1009 & 13\\
  23 & 2 & 1587 & 229 & 22 & 273 & 45\\
  23 & 2 & 1587 & 45 & 22 & 89 & 69\\
  23 & 2 & 1587 & 21 & 16 & 53 & 184\\
  23 & 3 & 36501 & 68 & 21 & 110 & 1518 \\
  23 & 3 & 36501 & 45 & 21 & 87 & 1587\\
  \bottomrule
  \end{tabular}
  \label{tab2}
  \end{threeparttable}
\end{table}

\begin{table}[tp]
  \scriptsize
  \centering
  \caption{Sample parameters for an $[n,k,d]$ cyclic code $C_{i}\curlyvee C_{j}$ with $i\in \{1,3\}$ and $j\in\{2,4\}$.}
  \setlength{\belowcaptionskip}{6pt}
  \renewcommand\arraystretch{1.3}
  \begin{threeparttable}
  \begin{tabular}{p{0.5cm}<{\centering}|p{0.5cm}<{\centering}p{0.5cm}<{\centering}p{0.5cm}<{\centering}p{0.5cm}<{\centering}p{0.5cm}<{\centering}p{0.5cm}<{\centering}p{0.5cm}<{\centering}p{0.5cm}<{\centering}p{0.5cm}<{\centering}p{0.5cm}<{\centering}p{0.5cm}<{\centering}}
  \toprule
  \textbf{Case} & $\bm{p}$ & $\bm{s}$ & $\bm{n}$ & $\bm{a_{i,0}}$ & $\bm{a_{i,1}}$ & $\bm{a_{j,0}}$ & $\bm{a_{j,1}}$ & $\bm{d_{i}}^{\dagger}$ & $\bm{d_{j}}$ & $\bm{d}$ &$\bm{k}$ \\
  \midrule
   1 & 5 & 2 & 150 & 19 & 18 & 19 & 4 & 3 & 9 & 6 & 82\\
   2 & 5 & 2 & 150 & 19 & 4 & 14 & 3 & 9 & 12 & 12 & 47 \\
   3 & 5 & 2 & 150 & 19 & 4 & 4 & 2 & 9 & 20 & 18 & 35 \\
   4 & 5 & 2 & 150 & 4 & 2 & 14 & 3 & 20 & 12 & 20 & 28 \\
   5 & 5 & 3 & 750 & 14 & 4 & 9 & 3 & 50 & 75 & 75 & 37 \\
   6 & 5 & 4 & 3750 & 19 & 4 & 9 & 3 & 225 & 375 & 375 & 42 \\
   7 & 11& 2 & 726 & 109 & 65 & 21 & 10 & 7 & 33 & 14 & 280 \\
   8 & 11 & 2 & 726 & 21 & 10 & 10 & 6 & 33 & 66 & 66 & 63 \\
   9 & 11 & 3 & 7986 & 65 & 10 & 32 & 9 & 231 & 330 & 330 & 135 \\
   10 & 23 & 2 & 3174 & 459 & 275 & 229 & 22 & 13 & 45 & 26 & 1282 \\
   11 & 23 & 2 & 3174 & 229 & 22 & 45 & 22 & 45 & 69 & 69 & 362 \\
   12 & 23 & 2 & 3174 & 229 & 22 & 21 & 16 & 45 & 184 & 90 & 326 \\
   13 & 23 & 2 & 3174 & 45 & 22 & 21 & 16 & 69 & 184 & 138 & 142 \\
  \bottomrule
  \end{tabular}
  \begin{tablenotes}
  \item[$\dagger$] $d_{i}$ and $d_{j}$ denote the minimum distances of  $C_{i}$ and  $C_{j}$, respectively.
  \end{tablenotes}
  \end{threeparttable}
  \label{tab3}
\end{table}

Combined with the results of Ref.~\cite{Lan2018Non} regarding a
 $3p^s$-length cyclic code $C_{i}=\langle (x-1)^{p^s-a_{i,0}}(x^2+x+1)^{p^s-a_{i,1}}\rangle$ with $a_{i,1}<a_{i,0}$, the minimum distance of a $6p^s$-length cyclic code $C_{i}\curlyvee C_{j}$ on $(\bm{u}+\bm{v}|\bm{u}-\bm{v})$ construction can thus be  determined. Sample parameters are provided in Table~\ref{tab3}.

 We can tell that in  the cases 3, 7, 8, 10, 12 and 13 in Table~\ref{tab3}, $C_{i}\curlyvee C_{j}$ have minimum distances twice as large as the component cyclic codes $C_{i}$, for $i\in\{1,3\}$ and $j\in \{2,4\}$. As a consequence, the quantum synchronizable codes derived from $C_{i}\curlyvee C_{j}$ can correct Pauli errors of weight twice larger than those constructed from  $C_{i}$.
In many instances, the former codes also have better error-correcting capability against Pauli errors than the quantum synchronizable codes derived from non-primitive narrow-sense BCH codes~\cite{Fujiwara2013Algebraic}.
 Denote by $\delta$ the precise lower bound~\cite{Lan2018Non,Aly2007On} for the  minimum distance of a dual-containing BCH code. Table~\ref{tab4} lists some sample parameters of dual-containing non-primitive, narrow-sense BCH codes.

\begin{table}[tp]
 \scriptsize
 \centering
  \caption{Sample parameters for $p$-ary non-primitive, narrow-sense BCH codes.}
  \setlength{\belowcaptionskip}{6pt}
  \renewcommand\arraystretch{1.3}
\begin{tabular}{c|cccc}
\toprule
\textbf{Case} & $\bm{p}$ & \textbf{length} & $\bm{\delta}$ &\textbf{dimension}\\
\midrule
1 & 5 & 146 & 5 & 130\\
2 & 5 & 748 & 28 & 638\\
3 & 11 & 725 & 60 & 563 \\
4 & 11 & 7985 & 65 & 7749\\
5 &¡¡23 & 3172 & 132 & 2794\\
\bottomrule
\end{tabular}
\label{tab4}
\end{table}

By the comparison between Tables~\ref{tab3} and~\ref{tab4} we can see that, given the same base field $\mathbb{F}_{p}$,  repeated-root cyclic codes $C_{i}\curlyvee C_{j}$, with well-chosen parameters, can  possess  larger minimum distances than non-primitive, narrow sense BCH codes of close lengths. In particular, over the base fields $\mathbb{F}_{11}$ and $\mathbb{F}_{23}$, repeated-root cyclic codes $C_{i}\curlyvee C_{j}$ of lengths  726 and 3174, respectively, can reach a larger minimum distance provided that parameters $a_{i,0},a_{i,1},a_{j,0},a_{j,1}$ are sufficiently small. In that case, quantum sychronizable codes constructed from $C_{i}\curlyvee C_{j}$ have better performance in correcting Pauli errors than those from non-primitive, narrow-sense BCH codes.

\section{\label{sec:level1}The product construction}

Apart from the $(\bm{u}+\bm{v}|\bm{u}-\bm{v})$ method, the product construction is another useful technique of generating new cyclic codes from old ones. Without loss of generality, we restrict the following discussion to the binary case.

Let $C_{1}$ and $C_{2}$ be linear codes of parameters $[n_{1},k_{1},d_{1}]$ and $[n_{2},k_{2},d_{2}]$ respectively. A product code $C=C_{1}\otimes C_{2}$~\cite{Blahut2003Algebraic} is defined to be an $[n_{1}n_{2},k_{1}k_{2},d_{1}d_{2}]$ linear code whose codewords are all the two-dimensional arrays where each row is a codeword in $C_{1}$ and each column is a codeword in $C_{2}$.
Denote by $c_{i,j}$ the element in the $(i+1)^{\text{th}}$ row and $(j+1)^{\text{th}}$ column of the array, where $0\leq i\leq n_{1}-1$ and $0\leq j\leq n_{2}-1$.
Then a codeword of $C$ can be identified with a bivariate polynomial $c(y,z)=\sum_{i=0}^{n_{1}-1}\sum_{j=0}^{n_{2}-1}c_{i,j}y^{i}z^{j}\in\frac{\mathbb{F}_{2}[x,y]}{\langle (y^{n_{1}}-1)(z^{n_{2}}-1)\rangle}$.
 According to the Chinese remainder theorem, there exists a unique integer $\theta$ in the range $0\leq \theta\leq n_{1}n_{2}-1$ such that $\theta\equiv i(\text{mod } n_{1})$ and $\theta\equiv j (\text{mod } n_{2})$, provided that $\text{gcd}(n_{1},n_{2})=1$. In that case,  ${c}(x)=\sum_{\theta=0}^{n_{1}n_{2}-1}c_{(\theta(\text{mod } n_{1}), \theta(\text{mod } n_{2}))}x^{\theta}$ is also a polynomial representation of code $C$.

 Suppose that $C_{1}$ and $C_{2}$ are both cyclic. Then $C$ is also cyclic since $x{c}(x)(\text{mod } x^{n_{1}n_{2}}-1)$, which corresponds to $yzc(y,z)(\text{mod } y^{n_{1}}-1)(\text{mod } z^{n_{2}}-1)$, is a codeword of $C$. Denote by $g_{1}(x)$ and $g_{2}(x)$ the  respective generator polynomial of $C_{1}$ and $C_{2}$. Then $C$ and the dual code $C^{\perp}$ have respective  generator polynomial~\cite{Lin1970Further}
  \begin{equation}
  \begin{array}{l}
  g_{C}(x)=\text{gcd}(g_{1}(x^{\beta n_{2}})g_{2}(x^{\alpha n_{1}}),x^{n_{1}n_{2}}-1),\\
   g^{\perp}_{C}(x)=\text{gcd}(g_{1}^{\perp}(x^{\beta n_{2}}),g_{2}^{\perp}(x^{\alpha n_{1}})),
   \end{array}
    \end{equation}
   where $\alpha,\beta$ are integers satisfying $\alpha n_{1}+\beta n_{2}=1$, and $g_{1}^{\perp}(x)$ and $g_{2}^{\perp}(x)$ represent the respective generator polynomial of $C_{1}^{\perp}$
and $C_{2}^{\perp}$.
  Clearly,
  $C$ is self-orthogonal if either $C_{1}$ or $C_{2}$ is self-orthogonal.
  Applying cyclic product codes to Theorem~\ref{thm1}, we can then obtain a broad family of quantum synchronizable codes as follows.

 \begin{Theorem}

  Let $C_{1}=\langle g_{1}(x)\rangle$ be a self-orthogonal $[n_{1},k_{1}]$ cyclic code and $C_{2}=\langle g_{2}(x)\rangle$ be an $[n_{2},k_{2}]$ cyclic code with
   $\alpha n_{1}+\beta n_{2}=1$, where $\alpha,\beta$ are integers.
   Suppose that $C_{3}$ is a cyclic code
   with
   a generator polynomial $g_{3}(x)=g_{2}(x)\rho(x)$, where $\rho(x)$ is a non-trivial polynomial such that  $\rho(0)=1$. Assume that $\emph{gcd}(g_{3}^{\perp}(x),\rho^{*}(x))=1$ where $\rho^{*}(x)$ denotes the reciprocal polynomial of $\rho(x)$. Then for any non-negative pair $a_{l},a_{r}$ such that $a_{l}+a_{r}<\emph{ord}(\emph{gcd}(g_{1}^{\perp}(x^{\beta n_{2}}),\rho^{*}(x^{\alpha n_{1}})))$, there exists an $(a_{l},a_{r})-[[n_{1}n_{2},n_{1}n_{2}-2k_{1}k_{2}]]$ quantum synchronizable code.
   \label{thm7}
 \end{Theorem}
  \begin{proof}
  The cyclic  product code $D=C_{1}\otimes C_{3}$ and its dual code $D^{\perp}$ have generator polynomials
  \begin{equation}
  \begin{array}{l}
 g_{D}(x)=\text{gcd}(g_{1}(x^{\beta n_{2}})g_{3}(x^{\alpha n_{1}}),x^{n_{1}n_{2}}-1),\\
 g_{D}^{\perp}(x)=\text{gcd}(g_{1}^{\perp}(x^{\beta n_{2}}),g_{3}^{\perp}(x^{\alpha n_{1}})),
  \end{array}
  \end{equation}
  respectively. Denote by $C$ the product code of $C_{1}$ and $C_{2}$. Then the dual codes $C^{\perp}$ and $D^{\perp}$ are  dual-containing  cyclic codes such that $C^{\perp}\subset D^{\perp}$.
  Note that
  \begin{equation}
  g_{2}^{\perp}(x)=\left(\frac{x^{n_{2}}-1}{g_{2}(x)}\right)^{*}=\left(\frac{x^{n_{2}}-1}{g_{3}(x)}\rho(x)\right)^{*}=g_{3}^{\perp}(x)\rho^{*}(x).
  \end{equation}
  Hence the quotient polynomial $f(x)$ of the generator polynomials of $C^{\perp}$ and $D^{\perp}$ is given by
  \begin{equation}
  f(x)=\frac{g_{C}^{\perp}(x)}{g_{D}^{\perp}(x)}=\frac{\text{gcd}(g_{1}^{\perp}(x^{\beta n_{2}}),g_{3}^{\perp}(x^{\alpha n_{1}})\rho^{*}(x^{\alpha n_{1}}))}{\text{gcd}(g_{1}^{\perp}(x^{\beta n_{2}}),g_{3}^{\perp}(x^{\alpha n_{1}}))}=\text{gcd}(g_{1}^{\perp}(x^{\beta n_{2}}),\rho^{*}(x^{\alpha n_{1}})).
  \end{equation}
  Apply $C^{\perp}$ and $D^{\perp}$ to Theorem~\ref{thm1}, we can then obtain the quantum synchronizable code with the desired parameters.

  \hfill$\square$
 \end{proof}

 Note that the constraint $\text{gcd}(g_{3}^{\perp}(x),\rho^{*}(x))=1$ is equivalent with
\begin{equation}
\text{gcd}(h_{3}^{*}(x),\rho^{*}(x))=\text{gcd}(h_{3}(x),\rho(x))=\text{gcd}\left(\frac{h_{2}(x)}{\rho(x)},\rho(x)\right)=1,
\end{equation}
where $h_{i}(x)$ denotes the parity check polynomial of $C_{i}$ for $i\in\{2,3\}$.
Hence the non-trivial polynomial $\rho(x)$ can  always be  found, provided that $h_{2}(x)$ has at least two irreducible factors. On that condition, a broad range of $n_{2}$-length cyclic codes can be applied to the construction in Theorem~\ref{thm7}, which, accordingly, widen the family of quantum synchronizable codes to a large extent. In particular, the range of parameters' selection for quantum synchronization coding is greatly enlarged considering that  lengths of previous quantum synchronizable codes, apart from those built on repeated-root cyclic codes, are of limited forms, e.g., $2^s-1$~\cite{Fujiwara2013Algebraic} and $\frac{2^{st}-1}{2^s-1}$~\cite{Fujiwara2014Quantum}, where $s,t$ are positive integers.

Furthermore, the quantum synchronizable codes obtained from cyclic product codes can also reach the maximum synchronization error tolerance if
\begin{equation}
\text{ord}(\text{gcd}(g_{1}^{\perp}(x^{\beta n_{2}}),\rho^{*}(x^{\alpha n_{1}})))=\text{ord}(\text{gcd}(h_{1}(x^{\beta n_{2}}),\rho(x^{\alpha n_{1}})))=n_{1}n_{2},
\end{equation}
where $h_{1}(x^{\beta n_{2}})=\frac{x^{\beta n_{1}n_{2}}-1}{g_{1}(x^{\beta n_{2}})}$. For example, let $C_{1}$ be a $[7,3]$ cyclic code with a generator polynomial
\begin{equation}
g_{1}(x)=(x+1)(x^3+x+1).
 \end{equation}
 The dual code $C_{1}^{\perp}$ has a generator polynomial
 $ g_{1}^{\perp}(x)=x^3+x+1$.
  Clearly, $C_{1}$ is a self-orthogonal code. Let $C_{2}$ be a $[15,11]$ cyclic code with a generator polynomial
  \begin{equation}
  g_{2}(x)=x^4+x^3+1.
   \end{equation}
   Therefore, $C=C_{1}\otimes C_{2}$ is an $[105,33]$ cyclic code.
Note that
\begin{equation}
h_{2}(x)=(x+1)(x^2+x+1)(x^4+x+1)(x^4+x^3+x^2+x+1).
 \end{equation}
 By choosing $\rho(x)$ to be $x^4+x+1$, we can then obtain a $[15,7]$ cyclic code $C_{3}$ with a generator polynomial
 \begin{equation}
 g_{3}(x)=(x^4+x^3+1)(x^4+x+1).
  \end{equation}
  In that case, the cyclic product code $D=C_{1}\otimes C_{3}$ is an $[105,21]$ code.

Since $(-2)\times 7+1\times 15=1$, we have
\begin{equation}
\begin{array}{l}
g_{1}^{\perp}(x^{15})=x^{45}+x^{15}+1,\\
\rho^{*}(x^{-14})=x^{-56}+x^{-42}+1=x^{-56}(x^{56}+x^{14}+1).
\end{array}
\end{equation}
Their greatest common divisor is
\begin{equation}
\text{gcd}(g_{1}^{\perp}(x^{15}),f^{*}(x^{-14}))=x^{12}+x^{10}+x^9+x^7+x^6+x^4+1,
\end{equation}
which is of order 105. Hence following Theorem~\ref{thm7}, we can build an $[[105,39]]$ quantum synchronizable code that can tolerate misalignment by up to 105 qubits.

\section{\label{sec:level1}Conclusions}

In this paper, we present two  families of quantum synchronizable codes from cyclic codes built on the $(\bm{u}+\bm{v}|\bm{u}-\bm{v})$ construction and the product construction. In the former case,
most existing quantum synchronizable codes that provide the highest tolerance against synchronization errors can be generalized to larger cases. In particular,
 repeated-root codes of length $lp^s$ have been thoroughly investigated in quantum synchronization coding and can provide a better performance in correcting Pauli errors than  non-primitive, narrow-sense BCH codes. In the latter case, the loose restrictions on the component cyclic codes ensure a large augmentation of available quantum synchronizable codes. Besides, their synchronization error tolerance can also reach the maximum under certain circumstances.

\bibliographystyle{unsrt}
\bibliography{bib}

\begin{thebibliography}{10}

\bibitem{Sklar2001Digital}
Bernard Sklar.
\newblock {\em Digital Communications: Fundamentals and Applications}.
\newblock Prentice Hall, Upper Saddle River, NJ, 2nd edition, 2001.

\bibitem{Bregni2002Synchronization}
Stefano Bregni.
\newblock {\em Synchronization of Digital Telecommunications Networks}.
\newblock John Wiley \& Sons, New York, U.S., 2002.

\bibitem{Nielsen2010Quantum}
Michael~A. Nielsen and Isaac~L. Chuang.
\newblock {\em Quantum Computation and Quantum Information}.
\newblock Cambridge University Press, 2010.

\bibitem{Lidar2013Quantum}
Daniel~A. Lidar and Todd~A. Brun.
\newblock {\em Quantum Error Correction}.
\newblock Cambridge University Press, Cambridge, U.K., 2013.

\bibitem{Fujiwara2013High}
Yuichiro Fujiwara and Vladimir~D. Tonchev.
\newblock High-rate self-synchronizing codes.
\newblock {\em IEEE Transactions on Information Theory}, 59(4):2328--2335,
  2013.

\bibitem{Polyanskiy2013Asynchronous}
Yury Polyanskiy.
\newblock Asynchronous communication: Exact synchronization, universality, and
  dispersion.
\newblock {\em IEEE Transactions on Information Theory}, 59(3):1256--1270,
  2013.

\bibitem{Fujiwara2013Block}
Yuichiro Fujiwara.
\newblock Block synchronization for quantum information.
\newblock {\em Physical Review A}, 87(2):109--120, 2013.

\bibitem{Fujiwara2013Algebraic}
Yuichiro Fujiwara, Vladimir~D. Tonchev, and Tony W.~H. Wong.
\newblock Algebraic techniques in designing quantum synchronizable codes.
\newblock {\em Physical Review A}, 88(1):162--166, 2013.

\bibitem{Fujiwara2014Quantum}
Yuichiro Fujiwara and Peter Vandendriessche.
\newblock Quantum synchronizable codes from finite geometries.
\newblock {\em IEEE Transactions on Information Theory}, 60(11):7345--7354,
  2014.

\bibitem{Xie2014Quantum}
Yixuan Xie, Jinhong Yuan, and Yuichiro Fujiwara.
\newblock Quantum synchronizable codes from augmentation of cyclic codes.
\newblock {\em Plos One}, 6(2):e14641, 2014.

\bibitem{Guenda2015Algebraic}
K.~Guenda, G.~G.~La Guardia, and T.~A. Gulliver.
\newblock Algebraic quantum synchronizable codes.
\newblock arXiv preprint: https://arxiv.org/abs/1508.05977, 2015.

\bibitem{xie2016Q}
Yixuan Xie, Lei Yang, and Jinhong Yuan.
\newblock $q$-ary chain-containing quantum synchronizable codes.
\newblock {\em IEEE Communications Letters}, 20(3):414--417, 2016.

\bibitem{Lan2018Non}
Lan Luo and Zhi Ma.
\newblock Non-binary quantum synchronizable codes from repeated-root cyclic
  codes.
\newblock {\em IEEE Transactions on Information Theory}, 64(3):1461--1470,
  2018.

\bibitem{Huffman2010Fundamentals}
Cary~W. Huffman and Vera Pless.
\newblock {\em Fundamentals of Error-Correcting Codes}.
\newblock Cambridge University Press, Cambridge, U.K., 2010.

\bibitem{Dinh2008On}
Hai~Q. Dinh.
\newblock On the linear ordering of some classes of negacyclic and cyclic codes
  and their distance distributions.
\newblock {\em Finite Fields and Their Applications}, 14(1):22--40, 2008.

\bibitem{Chen2014Repeated}
Bocong Chen, Hai~Q. Dinh, and Hongwei Liu.
\newblock Repeated-root constacyclic codes of length $lp^s$ and their duals.
\newblock {\em Finite Fields and Their Applications}, 177:60--70, 2014.

\bibitem{Dinh2013Structure}
Hai~Q. Dinh.
\newblock Structure of repeated-root constacyclic codes of length $3p^s$ and
  their duals.
\newblock {\em Finite Fields and Their Applications}, 313:983--991, 2013.

\bibitem{Chen2015Repeated}
Bocong Chen, Hai~Q. Dinh, and Hongwei Liu.
\newblock Repeated-root constacyclic codes of length $2l^m p^n$.
\newblock {\em Finite Fields and Their Applications}, 33:137--159, 2015.

\bibitem{Ozadam2009The}
Hakan {\"O}zadam and Ferruh {\"O}zbudak.
\newblock The minimum hamming distance of cyclic codes of length $2p^s$.
\newblock In {\em International Symposium on Applied Algebra, Algebraic
  Algorithms, and Error-Correcting Codes}, pages 92--100. Springer, 2009.

\bibitem{Zeh2015Decoding}
Alexander Zeh and Markus Ulmschneider.
\newblock Decoding of repeated-root cyclic codes up to new bounds on their
  minimum distance.
\newblock {\em Problems of Information Transmission}, 51(3):217--230, 2015.

\bibitem{Hughes2000Constacyclic}
Garry Hughes.
\newblock Constacyclic codes, cocycles and a $u+v|u-v$ construction.
\newblock {\em IEEE Transactions on Information Theory}, 46(2):674--680, 2000.

\bibitem{Ling2001On}
San Ling and P~Sole.
\newblock On the algebraic structure of quasi-cyclic codes {I}: Finite fields.
\newblock {\em IEEE Transactions on Information Theory}, 47(7):2751--2760,
  2001.

\bibitem{Macwilliams1977The}
Florence~Jessie Macwilliams and Neil James~Alexander Sloane.
\newblock {\em The Theory of Error-Correcting Codes}.
\newblock North-Holland Publishing Company, 1977.

\bibitem{Castagnoli1991On}
Guy Castagnoli, J.~L. Massey, P.~A. Schoeller, and Niklaus~Von Seemann.
\newblock On repeated-root cyclic codes.
\newblock {\em IEEE Transactions on Information Theory}, 37(2):337--342, 1991.

\bibitem{Aly2007On}
Salah~A. Aly, A.~Klappenecker, and P.~K. Sarvepalli.
\newblock On quantum and classical {BCH} codes.
\newblock {\em IEEE Transactions on Information Theory}, 53(3):1183--1188,
  2007.

\bibitem{Blahut2003Algebraic}
Richard~E. Blahut.
\newblock {\em Algebraic Codes for Data Transmission}.
\newblock Cambridge University Press, 2003.

\bibitem{Lin1970Further}
Shu Lin and Edward Weldon.
\newblock Further results on cyclic product codes.
\newblock {\em IEEE Transactions on Information Theory}, 16(4):452--459, 1970.

\end{thebibliography}

\end{document}